\documentclass{article}
\usepackage{tikz}
\usepackage[onehalfspacing]{setspace}
\usetikzlibrary{decorations.markings}
\usepackage[utf8]{inputenc}
\usepackage{amsfonts}
\usepackage[utf8]{inputenc}
\usepackage{amsmath}
\usepackage[sorting=none,maxbibnames=99, backend=bibtex]{biblatex}
\bibliography{Bibliography}
\usepackage[margin=1.6in]{geometry}
\usepackage{amsthm}
\usepackage{authblk}
\newtheorem{thm}{Theorem} 
\usepackage{tikz}
\usepackage{appendix}
\theoremstyle{definition}
\newtheorem*{defn}{Definition}

\newtheorem*{rmk}{Remark}

\newtheorem{prop}{Proposition}

\usepackage{amssymb}

\newcommand{\A}{\mathcal A}
\newcommand{\T}{\mathcal T}

\newcommand{\R}{\mathbb R}
\newcommand{\C}{\mathbb C}

\newcommand\norm[1]{\left\lVert#1\right\rVert}

\title{Spatial Decay of Kubo's Canonical Correlation Functions}
\author{Bowen Yang\\ byyang$@$caltech.edu}
\affil{\textit{California Institute of Technology, Pasadena, CA 91125} }

\usepackage{xcolor}

\begin{document}

\maketitle

\begin{abstract}
    
     Kubo's canonical correlation functions (canonical correlators) describe the static response of a system in equilibrium to infinitesimal local perturbations. Knowing their decay properties with respect to spatial distance is important for many theoretical and experimental applications. For a thermal state of a system with short-range interactions, we prove that any knowledge of the decay rate of ordinary correlators readily translates into that of canonical correlators. As the former have been extensively studied, our result can lead to many new results on the latter. Throughout the paper, we use the framework of infinite-volume quantum lattice model. However, the method we use is adaptable for other physical scenarios not considered in this paper.
\end{abstract}
\section{Introduction}

Condensed matter experiments often measure the linear response to a small local perturbation of a system in thermodynamic equilibrium in order to gain information about the unperturbed system. One of the simplest quantities to measure is the static response \cite{Kubo} to an infinitesimal perturbation $H\rightarrow H+\lambda B$ on the Hamiltonian by an observable $B$. The change in the equilibrium expectation value of an observable $A$ is determined by a function $\langle\langle A, B \rangle \rangle$ called Kubo's canonical correlation function. It is natural to measure the rate at which this function decays with respect to the spatial distance between $A$ and $B$. In the theoretical study of quantum many body systems, knowledge of this decay rate is important. For example, the assignment of a local temperature to a subsystem relies crucially on the exponential decay of canonical correlators \cite{kliesch2014locality}. Another example is the proof of the energy Bloch's theorem  \cite{kapustin2019absence} which assumes a fast enough decay rate as well. 

 A study of decay rates of canonical correlators in general quantum systems is, therefore, warranted. While there has been some progress for finite quantum spin system at high temperature (see \cite{kliesch2014locality}), our understanding of the topic is far from complete. On the other hand, a lot is known about ordinary correlation functions and their decay in thermal states. In 1969, Araki \cite{araki1969gibbs} proved the exponential and uniform clustering of ordinary correlators for one dimensional quantum spin chains at positive temperature. The clustering of ordinary correlators has been proved for certain higher dimensional spin systems, both classical \cite{bricmont1996high} and quantum \cite{NS}. It would be interesting to prove analogous results for canonical correlators. There could be two applications if such results exist. Firstly, it would justify the theoretical assumptions about rapid decay of canonical correlators away from phase transition. Secondly, near a continuous phase transition, ordinary correlators have power-like decay. The exponents of the decay rates are believed to be universal. A natural question is whether this translates into power-like decay with the same exponents for canonical correlators. The answer is affirmative for classical systems where ordinary and canonical correlators are exactly equal. It would be interesting to verify if this result holds also for quantum systems at positive temperature.

In this article we prove such results in the framework of infinite-volume quantum spin systems at positive temperature. We believe our methods can be adapted to many other discrete or continuous models. 
\\

\noindent\textbf{Acknowledgements:} I would like to express my gratitude towards Anton Kapustin for suggesting the problem and for his invaluable guidance and support. I also thank Nathaniel Sagman, Tamir Hemo and Alexandre Perozim de Faveri for discussions on decay rates.
\section{Setup and statement of the main result}
\subsection{Setup and notation}

We consider a quantum spin system defined on an infinite discrete metric space $(\Gamma, d)$. Let $P_0(\Gamma)$ be the set of finite subsets of $\Gamma.$ For a finite subset $X\in \mathcal P_0(\Gamma)$, we define $B_r(X) =\{y\in \Gamma: d(y,X)<r\}$. We also assume there is a constant $D>0$, such that for all $X \in \mathcal P_0(\Gamma)$ there exists $C>0$ with
\begin{equation} \label{polygrowth}
    |B_r(X)| \leq C|X|(1+r)^D.
\end{equation}
The algebra of local observables is 
\begin{equation}
    \A_{\text{loc}}= \varinjlim_{X\in \mathcal P_0(\Gamma)}\bigotimes_{x\in X}M_{d_x},
\end{equation}
where $M_{d_x}$ is the algebra of $d_x \times d_x$ matrices with $\sup_{x\in X}d_x<\infty$. We also denote $\bigotimes_{x\in X}M_{d_x}$ by $\A_X$ for all $X\in \mathcal{P}_0(\Gamma)$. The $C^*$-algebra $\A$ of quasi-local observables is the norm completion of $\A_{\text{loc}}$. For $A\in \A_{\text{loc}}$, its \textit{support} supp $A$ is the smallest $X\in \mathcal P_0(\Gamma)$ such that $A \in \A_X$. 

In a system of finite volume, time evolution is generated by a self-adjoint operator called the Hamiltonian. This is replaced by an interaction in an infinite-volume system. An interaction $\Phi$ is a function from $\mathcal P_0( \Gamma)$ to $\A$ such that $\Phi(X)=\Phi(X)^*\in \A_X$. In order to determine a time evolution on the infinite system, an interaction has to be reasonably local. Mathematically, this translates into a decay condition on $\lVert \Phi(X) \rVert$ as $X$ grows in size. Various viable conditions have been explored by others (see \cite{HK, NS, NS1, NS2}). For clarity we shall follow a single convention and assume for all $x\in \Gamma$ the following holds:
\begin{equation}\label{LRassumption}
    \sum_{Z\ni x}\norm{\Phi(Z)}|Z|\textnormal{exp}(\mu \textnormal{diam} (Z))\leq v/2<\infty,
\end{equation}
where diam$(Z)= \max_{y,z \in Z}d(y,z)$ for $Z\in \mathcal P_0(\Gamma)$.
Under this condition, we are able to determine an infinite-volume time evolution 
$\tau: \R \rightarrow \text{Aut}(\A)$ as  
 a strongly continuous one-parameter group of $C^*$-algebra automorphisms. Specifically, it is shown (see \cite{BR1, BR2, simon2014statistical, naaijkens2013quantum}) that for all $A$ in a dense subset of $\A$, the following limit is well-defined
 \begin{equation}
     \delta(A) = \lim_{|\Lambda|\rightarrow \infty}\sum_{X\subset \Lambda} [\Phi(X), A],
 \end{equation} where $\Lambda\in \mathcal P_0(\Gamma)$ grows to eventually include any point in $\Gamma$. Then $\tau_t$ is obtained on a dense subset by the exponentiation of $\delta$ through a Taylor series. The convergence of both $\delta$ and $\tau$ are shown in the references above. In particular, $\tau_t(A)$ is differentiable with respect to $t\in \R$ for all $A\in \A_{\text{loc}}$ (see lemma 3.3.5 in \cite{naaijkens2013quantum}). We will use this fact later.

When working with the algebra of observables (instead of the Hilbert space of pure states),  a state $\phi: \A \longrightarrow \C$ is a positive linear functional of norm 1, or equivalently $\phi$ is bounded with $\norm{\phi}=\phi(1)=1.$ This immediately implies that for any observable $A\in \A$, we have a bound $|\phi(A)|\leq \norm A.$

Given a time evolution $\tau$ and an inverse temperature $\beta \in [0, \infty]$, a state $\phi$ is called $\beta$-KMS \footnote{KMS stands for Kubo-Martin-Schwinger} if for all $A,B\in \A_{loc}$ there exists a function $F_{A,B}(t)$ holomorphic and bounded on the strip $S_\beta= \{z\in \C| 0\leq \mathrm{Im} z \leq \beta\}$  such that
\begin{equation}
    F_{A,B}(t) = \phi(A\tau_{t}(B))
\end{equation}
and 
\begin{equation}
    F_{A,B}(t+i\beta)= \phi(\tau_t(B)A),
\end{equation}
for all $t\in \R.$ The KMS condition implies that $\phi$ is invariant under the automorphisms $\tau$ \cite{BR2}.
 The KMS states are well studied as an infinite-volume thermal state at temperature $T=1/ \beta$. In particular, it has been shown \cite{powers1975existence} that such a state always exists in a quantum spin system for any $\beta$. More detailed exposition on the topic can be found in  \cite{BR2, simon2014statistical} . For a fixed $\beta$-KMS state the ordinary and Kubo's canonical correlators are respectively
\begin{equation}
    \langle A,B\rangle_\phi=\phi(AB)-\phi(A)\phi(B)
\end{equation}
 and
 \begin{equation}
     \langle \langle A,B\rangle\rangle_\phi=\frac{1}{\beta}\int_0^\beta F_{A,B}(ib)db-\phi(A)\phi(B).
 \end{equation}

\subsection{Main result}

Given a KMS state at positive temperatrue, we establish spatial decay for its canonical correlators assuming spatial decay for the ordinary correlators. Only two other ingredients are needed: approximate locality (see the next section) and the KMS condition. The large number of existing results on correlators decays often come in two forms: exponential clustering or power-like decay (see \cite{NS, NS1, NS2,araki1969gibbs, bricmont1996high, kliesch2014locality, W,}). We choose to use a general monotone decreasing function $g: \R_{\geq 0}\rightarrow \R_{\geq 0}$ as the decay rate instead. While this has the benefit of generality, it may appear opaque for applications. Thus, we also discuss our result for $g(l)= e^{-l/\xi}$ and $g(l)= l^{-n}$ with some $n, \xi >0$.

\begin{thm} \label{T2}
In a quantum spin system as defined above, let $\phi$ be a $\beta$-KMS state with $0 \leq\beta< \infty$. Given $X, Y\in \mathcal{ P}_0(\Gamma)$ with $l=d(X,Y)>0$, if for all $A\in \A_X$, $B\in \A_Y$, \begin{equation} \label{eq:minsupportclustering}
    |\langle A,B\rangle_\phi|\leq c\norm{A}\norm{B}\min\{|X|,|Y|\}g(l),
\end{equation}where $c$ is a constant,
then
 \begin{equation}
    |\langle \langle A,B\rangle\rangle_\phi|\leq c'\norm{A}\norm{B}\min\{|X|^2,|Y|^2\}g'(l),
\end{equation}
where $c'$ is another constant and $g'(l)$ is  $g(l/4)$ or $e^{-\mu l/2}$ whichever one has slower decay. In particular, if $g(l)= l^{-n}$ then $g'(l)= l^{-n}$. If $g(l)= e^{-l/\xi}$, then $g'(l) = e^{-l/\xi'}$ where $\xi'$ is another constant independent of $A$, $B$ and $l$. 
\end{thm}

\section{Approximate Locality}
Locality in a physical system broadly refers to a limit on the propagation speed of information. In relativistic systems, no information is allowed to travel faster than the speed of light. However, non-relativistic systems often do not possess such a sharp bound. Nevertheless, an approximate version of locality can emerge in a multitude of many body systems. In this section, we first define approximate locality for a quantum spin system. Then we state the Lieb-Robinson bounds which guarantees approximate locality for the spin systems considered in this article.

\begin{defn}
In a quantum spin system with time evolution $\tau$, we say the system is \textit{approximately local} if for any $A\in \A_X$ and $r>0$, there exists a family of operators $A^r(t) \in \A_{B_r(X)}$ differentiable with respect to $t \in \R$, such that $A^r(0)=A$ and
\begin{equation} \label{ApproxLoc}
    \norm{\tau_t(A)-A^r(t)}\leq c\norm{A}|X| e^{-2\mu r}(e^{v|t|}-1),
\end{equation}
where $c, \mu, v$ are positive absolute constants. 
\end{defn}

\begin{rmk}
The condition can be interpreted as a bound on information leaked out of an effective light-cone. In relativistic field theory, the condition above holds without any leakage.  
\end{rmk}

 It turns out a large class of non-relativistict quantum spin systems possess approximate locality. This fact was originally observed by Lieb and Robinson in \cite{LiebRobinson} and become known as the Lieb-Robinson bounds. There exists many improved versions of this result for example in \cite{H, HK, NS2, BHV} just to name a few. Below, we state a version discussed in \cite{H}.

\begin{thm} \label{LiebRobinson}
 In a quantum spin system on $(\Gamma,d)$, assume the interaction $\Phi$ is such that for all $x\in \Gamma$, the following holds:
\begin{equation} \label{LRassump}
    \sum_{Z\ni x}\norm{\Phi(Z)}|Z|\textnormal{exp}(\mu \textnormal{diam} (Z))\leq v/2<\infty,
\end{equation} 
for some positive constants $c, \mu$ and $v$. Let $A,B$ be operators supported on sets $X,Y$, respectively. Then, if $l=d(X,Y)>0$,
\begin{align}
    \norm{[B,\tau_t(A)]}\leq c\norm A \norm B \min(|X|, |Y|) e^{-\mu l}(e^{v|t|}-1). \label{lr}
\end{align}

\end{thm}
\begin{prop}
A spin system satisfying the Lieb-Robinson bounds possesses approximate locality and vice versa.
\end{prop}

\begin{proof}
($\Rightarrow$): Given quasi-local $A\in \A$ and $X\in \mathcal P_0(\Gamma)$, we define for an integer $n > 0$
\begin{equation} 
    \T_X^n(A)=\int_{B_n(X)\backslash X} dU UAU^\dagger,
\end{equation}
where the integral is over all unitary matrices of the form $U=\bigotimes_{x\in B_n(X)\backslash X} U_x$ with $U_x\in \A_{\{x\}}$ using the Haar measure. We also assume the Haar measure in the integral is normalized so that 
\begin{equation}
    \lVert \T_X^n(A) \rVert \leq \int_{B_n(X)\backslash X} dU \lVert UAU^\dagger\rVert=\norm{A}\int_{B_n(X)\backslash X} dU= \lVert A\rVert.
\end{equation}
Moreover, $\{\T_X^n(A)\}_n$ is a Cauchy sequence: Given $\varepsilon>0$, there exists $T \in \A_{\mathrm{loc}}$ such that  $\norm{T-A} <\varepsilon/2$. For any large $m, n$ such that supp $T \subset B_n(X)$ and $m > n,$
\begin{align} 
    &\lVert \T_X^m(A)-\T_X^n(A)\rVert\notag
    \\
    \leq & \int_{B_n(X)\backslash X} dU \lVert U(\T_{B_n(X)}^{m-n}(A)-A )U^\dagger \rVert \notag 
    \\
    = & \norm{\int_{B_m(X)\backslash B_n(X)} dW WAW^\dagger-T+T -A} \notag
    \\
    \leq & \norm{\int_{B_m(X)\backslash B_n(X)} dW W(A-T)W^\dagger} + \norm{T-A} < \varepsilon, 
\end{align}
where the last step relies on the fact $[T,W] = 0$ for all $W \in \A_{B_m(X)\backslash B_n(X)}$. Let $\T_X(A)= \lim_{n\rightarrow \infty}\T_X^n(A)$. It is easy to see that
\begin{equation} \label{bound}
    \lVert \T_X(A) \rVert \leq \lVert A\rVert.
\end{equation}
Thus $\T_X$ is a bounded linear operator on $\A$. 

We are now ready to show approximate locality. Given $A\in \A_X$ and $r>0$, we define 
\begin{equation} 
    A^r(t)=\T_{B_r(X)}(\tau_t(A)),
\end{equation}
and for $n > r$,
\begin{equation} 
    A^r_n(t)=\T^{n-r}_{B_r(X)}(\tau_t(A)),
\end{equation} Clearly $A^r(0) = A$.
From (\ref{bound}),
\begin{equation} 
    \lVert A^r(t) \rVert \leq  \lVert A\rVert.
\end{equation}
Since $U\tau_t(A)U^\dagger=\tau_t(A)+U[\tau_t(A),U^\dagger]$, theorem \ref{LiebRobinson} implies
\begin{equation} \label{preapproxlocal}
  \norm{\tau_t(A)-A^r_n(t)}\leq \int_{B_n(X)\backslash B_r(X)} dU\norm{[U^\dagger,\tau_t(A)]} \leq c\lVert A\rVert e^{-\mu r}(e^{v|t|}-1). 
\end{equation}
Let $n\rightarrow \infty$ on both sides, we have
\begin{equation} 
  \norm{\tau_t(A)-A^r(t)} \leq c\lVert A\rVert e^{-\mu r}(e^{v|t|}-1). 
\end{equation}
As $\tau_t(A)$ is differentiable with respect to $t$ for local $A\in \A_X$ and $\T_{B_r(X)}$ is bounded and linear, $A^r(t)=\T_{B_r(X)}(\tau_t(A))$ is also differentiable.

($\Leftarrow$): Given $A\in \A_X, B\in \A_Y$ with $l=d(X,Y)>0,$ we can take $r=l/2$ and deduce from approximate locality the following:
\begin{align}
    \norm{[B,\tau_t(A)]}\leq \norm{[B,\tau_t (A)-A^r(t)]}\leq c \norm{A} \norm{B} |Y| e^{-\mu l}(e^{v|t|}-1), 
\end{align} 
where we assume $\min(|X|,|Y|)=|Y|$ without loss of generality.
In the first step, we used the fact $[B,A^r(t)]=0$ for $r=l/2.$ 
\end{proof}

\section{Proof of Main Result}

By the KMS condition, $f(z):=F_{A,B}(z)$ is bounded and analytic on $S_\beta$.  
At positive temperature $\beta$ is finite, so the interval $[0,\beta]$ is compact. It is, therefore, sufficient to establish the desired decay for the integrand  $f(ib)-\phi(A)\phi(B)$. As long as the bound is uniform in $b \in [0, \beta]$, it survives after the integration. This approach would not apply to system at zero temperature where $\beta = \infty$.   

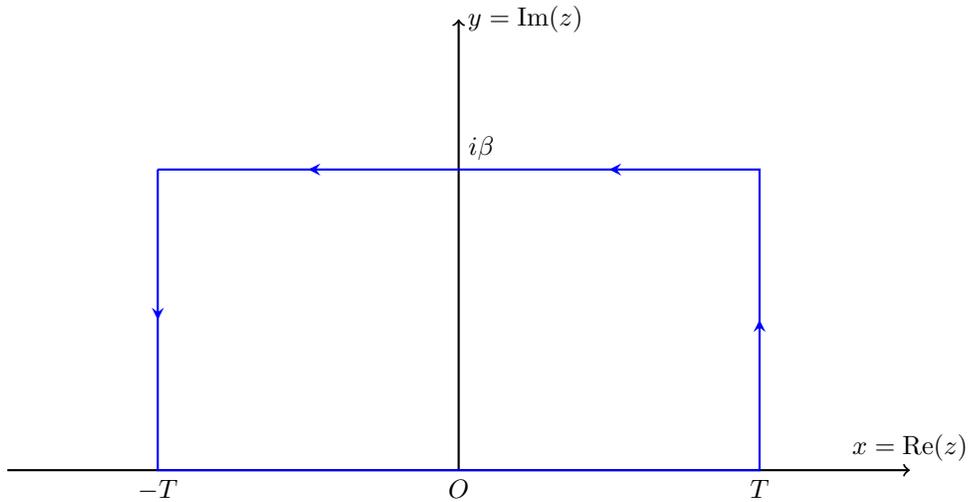
\begin{tikzpicture}[decoration={markings,
    mark=at position 2cm   with {\arrowreversed[line width=1pt]{stealth}},
    mark=at position 6cm   with {\arrowreversed[line width=1pt]{stealth}},
    mark=at position 10cm with {\arrowreversed[line width=1pt]{stealth}},
    mark=at position -2cm with {\arrowreversed[line width=1pt]{stealth}}
  }]
    
  \draw[thick, ->] (-6,0) -- (6,0) coordinate (xaxis);

  \draw[thick, ->] (0,0) -- (0,6) coordinate (yaxis);

  \node[above] at (xaxis) {$x=\mathrm{Re}(z)$};

  \node[right]  at (yaxis) {$y=\mathrm{Im}(z)$};

  \path[draw,blue, line width=0.8pt, postaction=decorate] 
        (-4,4)
    --  (4, 4)  node[midway, above right, black] {$i\beta$} 
    --  (4, 0)  node[midway, right, black] {} node[below, black] {$T$}
    --  (-4,0)  node[midway, below, black] {$O$} node[below, black] {$-T$} 
    --  (-4,4)  node[midway, left, black] {};
\end{tikzpicture}

\subsection{Contour Integration}

Fix $b \in [0, \beta]$ and let $w(z)=e^{-z^2-b^2}$. The function $w$ is bounded and analytic on the strip $S_\beta.$ Its exponent is conveniently chosen so that $w(ib) =1$.  
We integrate the function $\frac{f(z)w(z)}{z-ib}$ along the contour $\Gamma_T$ shown on the figure above. By the residue theorem,
\begin{align}
    2\pi i f(ib)w(ib)=&\int_{\Gamma_T} \frac{f(z)w(z)}{(z-ib)}dz\notag
    \\
    =&\int_{-T}^T \frac{f(t)w(t)}{t-ib}dx+\int_0^\beta\frac{f(T+iy)w(T+iy)}{T+iy-ib}dy\notag
    \\
    -&\int_{-T}^T\frac{f(t+i\beta)w(t+i\beta)}{t+i\beta-ib}dt-\int_0^\beta\frac{f(-T+iy)w(-T+iy)}{-T+iy-ib}dy. 
\end{align}
$f(z)$ and $w(z)$ are bounded on $S_\beta$, therefore, the 2nd and 4th terms vanish as $T\rightarrow \infty$.
\begin{align}
    f(ib)=&\frac{1}{2\pi i}\bigg(\int_{-\infty}^\infty \frac{f(t)w(t)}{t-ib}dt-\int_{-\infty}^\infty \frac{f(t+i\beta)w(t+i\beta)}{t+i\beta-ib}dt\bigg)\notag
    \\
    =&\frac{1}{2\pi i}\int_{-\infty}^\infty \frac{w(t)f(t)(t+i\beta-ib)-w(t+i\beta)f(t+i\beta)(t-ib)}{(t-ib)(t+i\beta-ib)}dt\notag
    \\
    =&\frac{1}{2\pi i}\int_{-\infty}^\infty \frac{w(t+i\beta)(f(t)-f(t+i\beta))}{t+i\beta-ib}dt\notag
    \\&+\frac{1}{2\pi i}\int_{-\infty}^\infty \frac{w(t)f(t)i\beta+(w(t)-w(t+i\beta))f(t)(t-ib)}{(t-ib)(t+i\beta-ib)}dt\notag
    \\
    =&\frac{1}{2\pi i}\int_{-\infty}^\infty \frac{w(t+i\beta)\phi([A,\tau_t(B)])}{t+i\beta-ib}dt\notag
    \\
    &+\frac{1}{2\pi i}\int_{-\infty}^\infty \frac{w(t)f(t)i\beta+(w(t)-w(t+i\beta))f(t)(t-ib)}{(t-ib)(t+i\beta-ib)}dt,
\end{align}
where in the last equality we used the fact $f(t+i\beta)=\phi(\tau_{t}(B)A).$ 
If we substitute $A=I$, the above identity reduces to 
\begin{equation}
    \frac{1}{2\pi i}\int_{-\infty}^\infty \frac{w(t)\phi(B)i\beta+(w(t)-w(t+i\beta))\phi(B)(t-ib)}{(t-ib)(t+i\beta-ib)}dt=\phi(B),
\end{equation}
or
\begin{equation}
    \frac{1}{2\pi i}\int_{-\infty}^\infty \frac{w(t)i\beta+(w(t)-w(t+i\beta))(t-ib)}{(t-ib)(t+i\beta-ib)}dt=1,
\end{equation}
This identity can be verified more directly using the residue theorem bearing in mind that $w(ib)=1$. With this identity we write

\begin{align}
    &f(ib)-\phi( A )\phi( B )\notag
    \\
    =&\frac{1}{2\pi i}\int_{-\infty}^\infty \frac{w(t+i\beta)\phi([A,\tau_t(B)])}{t+i\beta-ib}dt\notag
    \\
    &+\frac{1}{2\pi i}\int_{-\infty}^\infty \frac{w(t)(f(t)-\phi(A)\phi(B))i\beta+(w(t)-w(t+i\beta))(f(t)-\phi(A)\phi(B))(t-ib)}{(t-ib)(t+i\beta-ib)}dt\notag
    \\
=&\frac{1}{2\pi i}\int_{-\infty}^\infty \frac{w(t+i\beta)\phi([A,\tau_t(B)])}{t+i\beta-ib}dt\notag
\\
&+\frac{1}{2\pi i}\int_{-\infty}^\infty \frac{\langle A,\tau_t(B)\rangle_\phi(w(t) i\beta+(w(t)-w(t+i\beta))(t-ib))}{(t-ib)(t+i\beta-ib)}dt.
\end{align}

Therefore, 

\begin{align}
 2\pi\lvert&\phi( \tau_{-ib}(A)B )-\phi( A )\phi( B )\rvert\notag
 \\
\leq & \bigg|\int_{-\infty}^\infty \frac{w(t+i\beta)\phi([A,\tau_t(B)])}{t+i\beta-ib}dt\bigg|\notag
\\
&+\bigg|\int_{-\infty}^\infty \frac{\langle A,\tau_t(B)\rangle_\phi w(t)}{(t-ib)}dt\bigg|\notag
\\
&+\bigg|\int_{-\infty}^\infty \frac{\langle A,\tau_t(B)\rangle_\phi w(t+i\beta)}{(t+i\beta-ib)}dt\bigg|.
\end{align}
We proceed to bound each term separately. Throughout the next sections, we absorb all constant coefficients into a symbol $c$. Thus $c$ may change from step to step.

\subsection{First Term}
The decay for this term relies completely on approximate locality:
\begin{align}
&\bigg|\int_{-\infty}^\infty \frac{w(t+i\beta)\phi([A,\tau_t(B)])}{t+i\beta-ib}dt\bigg|\notag
\\
\leq & \int_{-\mu l/2v}^{\mu l/2v}\frac{|w(t+i\beta)|\norm{[A,\tau_t(B)]}}{|t+i\beta-ib|}dt\notag
\\
&+ \int_{|t|>\mu l/2v}\frac{|w(t+i\beta)|\norm{[A,\tau_t(B)]}}{|t+i\beta-ib|}dt\notag
\\
\leq &  \int_{-\mu l/2v}^{\mu l/2v}\frac{|w(t+i\beta)|\norm{[A,\tau_t(B)]}}{|t|}dt\notag
\\
&+2\norm{A}\norm{B} \int_{|t|>\mu l/2v}\frac{|w(t+i\beta)|}{|t|}dt
\end{align}
Finally we substitute $|w(t+i\beta)|=e^{-t^2+\beta^2-b^2}$ and then use approximate locality (or the Lieb-Robinson bounds) to show exponential decay:
\begin{align}
    &e^{\beta^2-b^2}\int_{-\mu l/2v}^{\mu l/2v}\frac{\norm{[A,\tau_t(B)]}}{|t|}e^{-t^2}dt\notag
    \\
    \leq& c\norm{A}\norm{B}\min(|X|,|Y|)e^{\beta^2-\mu l}\int_{0}^{\mu l/2v} \frac{e^{vt}-1}{t}e^{-t^2}dt\notag
    \\
    \leq& c\norm{A}\norm{B}\min(|X|,|Y|)e^{\beta^2-\mu l}\bigg(\int_{0}^{1} \frac{e^{vt}-1}{t}e^{-t^2}dt+\int_{1}^{\mu l/2v} e^{vt}dt\bigg)\notag
    \\
    =&c\norm{A}\norm{B}\min(|X|,|Y|)e^{\beta^2-\mu l}\bigg(\int_{0}^{1} \frac{e^{vt}-1}{t}e^{-t^2}dt+\frac{e^{\mu l/2}}{v}+\frac{e^v}{v}\bigg)\notag
    \\
    \leq& c \norm{A}\norm{B}\min(|X|,|Y|)e^{-\mu l/2}.
\end{align} 
The second integral becomes
\begin{align}
    &4\norm{A}\norm{B} \int_{t>\mu l/2v}\frac{e^{\beta^2-b^2-t^2}}{t}dt \notag
    \\
    \leq & c\norm{A}\norm{B}\int_{t>\mu l/2v}\frac{e^{-t^2}}{t}dt .
\end{align}
This decays exponentially in $l^2$ which is faster than the previous integral, so overall
\begin{equation}
    \bigg|\int_{-\infty}^\infty \frac{w(t+i\beta)\phi([A,\tau_t(B)])}{t+i\beta-ib}dt\bigg| \leq c \norm{A}\norm{B}\min(|X|,|Y|)e^{-\mu l/2}.
\end{equation}
Note this term decays exponentially with $l$ regardless of the rate of decay assumed for the ordinary correlator. The bound obtained is manifestly independent of $b$.
\subsection{Second and Third Terms}
Bounding the second and third terms similar, so we only provide the detail of the former:
\begin{align}
    &\bigg|\int_{-\infty}^\infty \frac{\langle A,\tau_t(B)\rangle_\phi w(t)}{t-ib}dt\bigg|\notag
    \\
    \leq &\bigg|\int_{-\mu l/2v}^{\mu l/2v}\frac{\langle A,\tau_t(B)\rangle_\phi w(t)}{t-ib}dt\bigg|\notag
    \\
    &+2\norm A\norm B\int_{|t|>\mu l/2v}\frac{|w(t)|}{|t|}dt.
\end{align}
We have used the simple bound $|\langle A,\tau_t(B)\rangle_\phi|=|\phi(A\tau_t(B))-\phi(A)\phi(B)|\leq \norm{A\tau_t(B)}+\norm A \norm B\leq 2\norm A\norm B$ for the last step. 

Since $w(t)=e^{-t^2-b^2}$, the second integral clearly decays exponentially with respect to $l^2$.

 By approximate locality, let $B^{3l/4}(t)$ be the local approximation to $\tau_t (B)$. Then
\begin{align}
    &\bigg|\int_{-\mu l/2v}^{\mu l/2v} \frac{w(t)}{t-ib}\langle A,\tau_t(B)\rangle_\phi dt\bigg|\notag
    \\
    = &\bigg| \int_{-\mu l/2v}^{\mu l/2v}\frac{w(t)}{t-ib}\big(\langle A,B^{3l/4}(t)\rangle_\phi+\phi(A(\tau_t(B)-B^{3l/4}(t)))-\phi(A)\phi((\tau_t(B)-B^{3l/4}(t)))\big)dt\bigg|\notag
    \\
    \leq &\bigg|\int_{-\mu l/2v}^{\mu l/2v}\frac{w(t)}{t-ib}\langle A,B^{3l/4}(t)\rangle_\phi dt\bigg|+2\norm{A}\int_{-\mu l/2v}^{\mu l/2v}\frac{|w(t)|}{|t-ib|}\norm{\tau_t(B)-B^{3l/4}(t)}dt.
\end{align}
By approximate locality, the second integral is bounded as follows,
\begin{align}
    &2\norm{A}\int_{-\mu l/2v}^{\mu l/2v}\frac{|w(t)|}{|t-ib|}\norm{\tau_t(B)-B^{3l/4}(t)}dt\notag
    \\
    \leq &2\norm{A}\int_{-\mu l/2v}^{\mu l/2v}\frac{e^{-t^2}}{|t|}\norm{\tau_t(B)-B^{3l/4}(t)}dt\notag
    \\
    \leq&c\norm{A}\norm{B}\min(|X|,|Y|)e^{-3\mu l/2} \int_{0}^{\mu l/2v}\frac{e^{vt}-1}{t}e^{-t^2}dt\notag
    \\
    \leq&c\norm{A}\norm{B}\min(|X|,|Y|)e^{-3\mu l/2} \bigg(\int_{0}^{1} \frac{e^{vt}-1}{t}e^{-t^2}dt+\frac{e^{\mu l/2}}{v}+\frac{e^v}{v}\bigg)\notag
    \\
    \leq&c\norm{A}\norm{B}\min(|X|,|Y|)e^{-\mu l} .
\end{align}
We are now left with 
\begin{equation}
    \bigg|\int_{-\mu l/2v}^{\mu l/2v}\frac{e^{-t^2-b^2}}{t-ib}\langle A,B^{3l/4}(t)\rangle_\phi dt\bigg|
\end{equation} 
whose decay rate depends on $g$ the decay rates of  ordinary correlators.

By (\ref{eq:minsupportclustering}), we have
  \begin{equation} \label{eq:minsupportclustering'}
     |\langle A,B^{3l/4}(t)\rangle_\phi|\leq c\min(|X|,|Y|)\norm{A}\norm{B}g(l/4),
 \end{equation}
\begin{rmk}
We have $\min(|X|,|Y|)$ above instead of $|X|$ because we could have switched the roles of $A$ and $B$ in the entire proof.
\end{rmk}
Now, we fix an $\epsilon>0$ and estimate
\begin{align}\label{eq:epsestimate}
    &\bigg|\int_{-\mu l/2v}^{\mu l/2v}\frac{e^{-t^2-b^2}}{t-ib}\langle A,B^{3l/4}(t)\rangle_\phi dt\bigg|\notag
    \\
    =&\bigg|\int_{0}^{\mu l/2v}\frac{e^{-t^2-b^2}}{t-ib}\langle A,B^{3l/4}(t)\rangle_\phi -  \frac{e^{-t^2-b^2}}{t+ib}\langle A,B^{3l/4}(-t)\rangle_\phi dt\bigg|\notag
    \\
    \leq& \bigg|\int_0^\epsilon\frac{e^{-t^2-b^2}}{t-ib}\langle A,B^{3l/4}(t)\rangle_\phi -  \frac{e^{-t^2-b^2}}{t+ib}\langle A,B^{3l/4}(-t)\rangle_\phi dt \bigg|\notag
    \\
    &+\bigg|\int_{\epsilon}^{\mu l/2v}\frac{e^{-t^2-b^2}}{t-ib}\langle A,B^{3l/4}(t)\rangle_\phi -  \frac{e^{-t^2-b^2}}{t+ib}\langle A,B^{3l/4}(-t)\rangle_\phi dt\bigg|\notag
    \\
    \leq& \bigg|\int_0^\epsilon\frac{e^{-t^2-b^2}}{t-ib}\langle A,B^{3l/4}(t)\rangle_\phi -  \frac{e^{-t^2-b^2}}{t+ib}\langle A,B^{3l/4}(-t)\rangle_\phi dt \bigg|\notag
    \\
    &+\int_{\epsilon}^{\infty}\frac{e^{-t^2}}{t}|\langle A,B^{3l/4}(t)\rangle_\phi |+  \frac{e^{-t^2}}{t}|\langle A,B^{3l/4}(-t)\rangle_\phi |dt.
\end{align}
To get the desired decay of the last term in eq. (\ref{eq:epsestimate}),  we apply  (\ref{eq:minsupportclustering'}) to $|\langle A,B^{3l/4}(t)\rangle_\phi |$ as well as to $|\langle A,B^{3l/4}(-t)\rangle_\phi |$. Note $\int_\epsilon^{\infty}\frac{e^{-t^2}}{t}dt$ is a finite constant for a fixed $\epsilon$, so
\begin{equation}
    \int_{\epsilon}^{\infty}\frac{e^{-t^2}}{t}|\langle A,B^{3l/4}(t)\rangle_\phi |+  \frac{e^{-t^2}}{t}|\langle A,B^{3l/4}(-t)\rangle_\phi |dt\leq c\min(|X|,|Y|)\norm{A}\norm{B}g(l/4).
\end{equation}
Finally, we are left with the terms
\begin{align}
    &\bigg|\int_0^\epsilon\frac{e^{-t^2-b^2}}{t-ib}\langle A,B^{3l/4}(t)\rangle_\phi -  \frac{e^{-t^2-b^2}}{t+ib}\langle A,B^{3l/4}(-t)\rangle_\phi dt \bigg|\notag
    \\
    \leq& \int_0^\epsilon e^{-t^2-b^2}\frac{|\langle A,B^{3l/4}(t)\rangle_\phi(t+ib)-\langle A,B^{3l/4}(-t)\rangle_\phi(t-ib)|}{t^2+b^2}dt\notag
    \\
    \leq & \int_0^\epsilon\frac{|\langle A,B^{3l/4}(t)\rangle_\phi|b+|\langle A,B^{3l/4}(-t)\rangle_\phi| b}{t^2+b^2}dt\notag
    \\
    &+\int_0^\epsilon\frac{t|\langle A,B^{3l/4}(t)\rangle_\phi-\langle A,B^{3l/4}(-t)\rangle_\phi|}{t^2+b^2}dt
\end{align}
To the first integral, apply (\ref{eq:minsupportclustering'}) again while noting that the remaining integral $\int_0^\epsilon\frac{b}{t^2+b^2}dt$ gives $\arctan(\epsilon/b)$ which is bounded for any $b$. Therefore,
\begin{equation}
    \int_0^\epsilon\frac{|\langle A,B^{3l/4}(t)\rangle_\phi|b+|\langle A,B^{3l/4}(-t)\rangle_\phi| b}{t^2+b^2}dt \leq c\min(|X|,|Y|)\norm{A}\norm{B}g(l/4).
\end{equation}
To estimate the final integral, recall that  $B^{3l/4}(t)$ is differentiable with respect to $t$. By the mean value theorem,
\begin{align}
    &\int_0^\epsilon\frac{t|\langle A,B^{3l/4}(t)\rangle_\phi-\langle A,B^{3l/4}(-t)\rangle_\phi|}{t^2+b^2}dt  \notag\\ \leq&\int_0^\epsilon\frac{|\langle A,B^{3l/4}(t)\rangle_\phi-\langle A,B^{3l/4}(-t)\rangle_\phi|}{t}dt\\
    \leq &4\epsilon \sup_{-\epsilon<t<\epsilon}\big|\frac{d}{dt'}\bigg|_{t'=t}\langle A,B^{3l/4}(t')\rangle_\phi\big|=4\epsilon \sup_{-\epsilon<t<\epsilon}\big|\langle A,\frac{dB^{3l/4}(t')}{dt'}\bigg|_{t'=t}\rangle_\phi\big|.
\end{align}
 Since $B^{3l/4}(t')$ is supported on $B^{3l/4}(Y)$ for all $t'$, the support of $\frac{dB^{3l/4}(t')}{dt'}\big|_{t'=t}$ is also contained in $B^{3l/4}(Y)$. Next we bound its norm, 
\begin{align}
    &\norm{\frac{dB^{3l/4}(t')}{dt'}\big|_{t'=t}}
    \leq \norm{\frac{d}{dt'}\bigg|_{t'=0}\tau_{t+t'}(B)}
    =\norm{\lim_{|\Lambda|\rightarrow \infty}\sum_{Z\subset \Lambda}[\Phi(Z),\tau_{t}(B)]}\notag
    \\
    \leq & \lim_{|\Lambda|\rightarrow \infty}\norm{\sum_{Z\subset \Lambda: Z\cap Y\ne \emptyset}[\Phi(Z),\tau_{t}(B)]}+\lim_{|\Lambda|\rightarrow \infty}\norm{\sum_{Z\subset \Lambda: Z\cap Y=\emptyset}[\Phi(Z),\tau_{t}(B)]}.
\end{align}
We bound the former term using (\ref{LRassumption}),
\begin{align}
    \lim_{|\Lambda|\rightarrow \infty}\norm{\sum_{Z\subset \Lambda: Z\cap Y\ne \emptyset}[\Phi(Z),\tau_{t}(B)]}
   \leq 2\norm{B}\sum_{y\in Y}\sum_{Z\ni y}\norm{\Phi(Z)}\leq v|Y|\norm{B}.
\end{align}
As for the latter, apply (\ref{lr}) and then (\ref{LRassumption})
\begin{align}
    &\lim_{|\Lambda|\rightarrow \infty}\norm{\sum_{Z\subset \Lambda: Z\cap Y=\emptyset}[\Phi(Z),\tau_{t}(B)]}\notag
    \\
    \leq & \sum_{Z\cap Y=\emptyset}\norm{[\Phi(Z),\tau_{t}(B)]}\notag
    \\
    \leq  & \sum_{Z\cap Y=\emptyset} 2\norm{\Phi(Z)}\norm{B}|Z|e^{-\mu d(Z,Y)}(e^{v|t|}-1)\notag
    \\
    \leq & 2\norm{B}(e^{v\epsilon}-1)\sum_{Z\cap Y=\emptyset} \norm{\Phi(Z)}|Z|e^{-\mu d(Z,Y)}\notag
    \\
    \leq & 2\norm{B}(e^{v\epsilon}-1)\sum_{r=1}^\infty e^{-\mu r} \sum_{ d(z,Y)\in (r-1,r]}\sum_{Z\ni z}\norm{\Phi(Z)}|Z|\notag
    \\
    \leq & v\norm{B}(e^{v\epsilon}-1)\sum_{r=1}^\infty e^{-\mu r} D(r),
\end{align}
where $|t|<\epsilon$ and $D(r):=|\{z\in \Gamma: d(z,Y)\in (r-1,r]\}|$. Clearly, $D(r)$ grows polynomially with respect to $r$, so $\sum_{r=1}^\infty e^{-\mu r} D(r)$ is bounded by a constant linearly dependent on $D(1)$. This is dominated by $|Y|$ due to (\ref{polygrowth}). Overall, \begin{equation} \label{overall}
    \norm{\frac{dB^{3l/4}(t')}{dt'}\big|_{t'=t}}<C(\mu,v,\epsilon)\norm{B}|Y|.
\end{equation}
Then we apply   (\ref{eq:minsupportclustering})  to $\big|\langle A,\frac{dB^{3l/4}(t')}{dt'}\big|_{t'=t}\rangle_\phi\big|$ and obtain
\begin{equation}
    \int_0^\epsilon\frac{t|\langle A,B^{3l/4}(t)\rangle_\phi-\langle A,B^{3l/4}(-t)\rangle_\phi|}{t^2+b^2}dt\leq c\norm{A}\norm{B}\min(|X|,|Y|)|Y|g(l/4).
\end{equation}

Furthermore, the canonical correlator is symmetric, so if we swap $A$ and $B$ from the very beginning, we get the same overall bound except an additional factor of $|X|$ instead of $|Y|$. Thus, we actually have
\begin{equation}
    \int_0^\epsilon\frac{t|\langle A,B^{3l/4}(t)\rangle_\phi-\langle A,B^{3l/4}(-t)\rangle_\phi|}{t^2+b^2}dt\leq c\norm{A}\norm{B}\min(|X|^2,|Y|^2)g(l/4).
\end{equation}

The procedure for bounding the third term is entirely similar.
\subsection{Summary}
In summary, \begin{equation}
    |\langle \langle A,B\rangle\rangle_\phi|\leq c\norm{A}\norm{B}\min\{|X|^2,|Y|^2\}g'(l),
\end{equation} where $g'(l)$ is the slower decaying one between $g(l/4)$ and $e^{-\mu l/2}$.
\section{Discussion}
In the present paper, we focus on quantum spin systems with short-range interactions. In particular, interactions are assumed to satisfy a concrete condition (\ref{LRassumption}). However, our analysis can be adapted readily if we substitute (\ref{LRassumption}) with many other conditions for short-range interactions such as those analyzed in \cite{NS, HK}. Lieb-Robinson bounds have also been proven for certain spin systems with long-range interactions (see \cite{NS, HK, matsuta2017improving}). However, these bound may not lead to approximate locality defined here. For example, \cite{matsuta2017improving} demonstrates that systems with long-range interactions could have unbounded information propagation speeds. Therefore, it is interesting to investigate decay rates of canonical correlators in systems with long-range interactions. 

\printbibliography

\end{document}